\documentclass[a4paper, 12pt]{article}
\usepackage{amsmath,amsthm,amssymb}
\usepackage{graphicx}
\usepackage{color}
\usepackage{dsfont}
\usepackage{epstopdf}
\usepackage[hang]{subfigure}
\usepackage{natbib}
\usepackage{ulem}

\newtheorem{thm}{Theorem}
\newtheorem{lem}{Lemma}

\newtheorem{prop}{Proposition}
\newtheorem{assumption}{Assumption}
\newcounter{rmk}

\newcommand\e{\epsilon}
\newcommand\lam{\lambda}
\newcommand\mR{\mathds{R}}
\newcommand\mN{\mathcal{N}}
\def\t{{ \mathrm{\scriptscriptstyle T} }}

\newcommand\sign[1]{\textnormal{sign}{(#1)}}
\newcommand\en{\textnormal{en}}

\title{Sparse Laplacian Shrinkage with the Graphical Lasso Estimator for Regression Problems}
\author{Yuehan Yang\thanks{School of Statistics and Mathematics, Central University of Finance and Economics},
~~ Siwei Xia$^\dag$~~and~~Hu Yang\thanks{College of Mathematics and Statistics, Chongqing University}
}
\date{}

\begin{document}
\maketitle
\begin{abstract}
This paper considers a high-dimensional linear regression problem where there are complex correlation structures among predictors. We propose a graph-constrained regularization procedure, named Sparse Laplacian Shrinkage with the Graphical Lasso Estimator (SLS-GLE). The procedure uses the estimated precision matrix to describe the specific information on the conditional dependence pattern among predictors, and encourages both sparsity on the regression model and the graphical model. We introduce the Laplacian quadratic penalty adopting the graph information, and give detailed discussions on the advantages of using the precision matrix to construct the Laplacian matrix. Theoretical properties and numerical comparisons are presented to show that the proposed method improves both model interpretability and accuracy of estimation. We also apply this method to a financial problem and prove that the proposed procedure is successful in assets selection.
\end{abstract}

\vspace*{4mm}

\section{Introduction}
High-dimensional regression problems have received continuous interest and tremendous effort to develop new methodologies and theories. There are several important methods proposed to treat the task, such as the Lasso of \citet{tibshirani1996lasso}, SCAD of \citet{fan2001variable}, MCP of \citet{zhang2010mcp}, etc. These methods are able to estimate and select variables simultaneously and effectively, especially in high-dimensional settings.
Many works have been made in studying the computational algorithms and theoretical properties of these methods, i.e. \citet{friedman2007pathwise,friedman2010regularization,fan2004nonconcave,zhao2006lasso,wainwright2009sharp,mazumder2011sparsenet}.

A usual assumption for above methods is the requirement of independence between the penalty and correlation among predictors, and yet, correlated features become more and more common in many applications. To tackle this problem, several efficient methods have been proposed. Elastic Net, proposed by \citet{zou2005elastic}, utilizes the $l_1$ and $l_2$ penalty to select variables with grouping effect. Group lasso \citep{yuan2006model} and the Sparse-group lasso \citep{simon2013sparse} incorporate group or hierarchical structure of variables.
The Mnet method proposed from \citet{huang2016mnet}, a combination of the MCP and $l_2$ penalty, was shown to be equal to the oracle ridge estimator with high probability under reasonable conditions. \citet{tibshirani2005fused} proposed a new penalty, called Fused Lasso, to promote smoothness over ordered variables and successfully maintain grouping effects. Inspired by the fused penalty, \citet{she2010cluster}, \citet{hebiri2011smooth} and \citet{guo2016spline} proposed the clustered Lasso, Smooth-Lasso and Spline-lasso, respectively.

However, there are still gaps in estimating correlated features. Rare methods focus on estimating the specific information on the correlated patterns among the predictors while it could greatly improve the accuracy of both estimations and predictions. \citet{li2008network,li2010graph} proposed a network-constrained and graph-constrained estimation procedure, combining the $ l_1 $ penalty and a Laplacian matrix penalty assumed to be known a priori network or graphical information. \citet{daye2009fusion} and \citet{huang2011laplacian} proposed the Weighted fusion and the Sparse Laplacian Shrinkage (SLS), respectively. Both methods use sample correlation between predictors to construct the Laplacian matrix penalty. There's one shortcoming that the sample correlation is usually insufficient in sparsity. To fix this problem, \citet{huang2011laplacian} applied a hard threshold function calculated by the Fisher transformation, and provided several forms of adjacency matrix for the Laplacian quadratic penalty.

Our goal in this paper is to study the data with complex patterns. We plan to measure the conditional dependencies among predictors and introduce the information of graph structure into the shrinkage penalty in order to better predict the regression model, furthermore, provide suitable procedure and algorithm achieving this target with theoretical support and full comparisons with other methods. When $ p \gg n $, we expect to capture low dimensional structures in both regression model and graphical model, and these sparse structures could help us focus on the important features. In light of this, we propose a new method, called Sparse Laplacian Shrinkage with the Graphical Lasso Estimator (SLS-GLE).
The procedure uses the Laplacian quadratic penalty and applies the estimated precision matrix to construct the graph Laplacian, then use the regular penalty for selection and sparsity of regression features. The contributions of this paper are as follows.
\begin{itemize}
\item[(1)] We measure the specific information between predictors and apply them into the proposed procedure. The graph structure gives us an insight into the complex patterns, and helps us achieving accurate modeling.
\item[(2)] In SLS-GLE, we use the conditional dependences to construct the graph Laplacian. As a contrast, SLS used the information related to the Pearson's correlation coefficient and combine with a hard thresholding function. Although the latter can obtain a sparse graph too, but easily leading to unnecessary bias and computational cost for the problems of large dimensionality.
\item[(3)] Furthermore, we expect the proposed method to describe different kinds of structures involved in predictors, including different sparsity patterns. We can easily solve this problem by adapting the tuning parameter of the Graphical lasso (Glasso). In the meantime, SLS needs to find a suitable form of adjacency measure for each data example.
\item[(4)] Finally, SLS-GLE uses the Glasso to estimate the graph structure. Glasso has remarkable advantages on the theoretical support and computational algorithm, leading to a nice estimation with high probability converging to the true precision matrix. It will support SLS-GLE to achieve its own oracle properties.
\end{itemize}

This paper is organized as follows. Section 2 presents the
method. Section 3 shows the full comparisons between SLS-GLE and other penalties. In Section 4 we study the theoretical properties of the SLS-GLE estimator. The simulations and application in Section 5 analyse the performance of the proposed method and compare with several existing methods. We conclude in Section 6. Technique details are provided in the Appendix.

\section{Method}
Consider a linear regression model with $ n $-dimensional response $ y $, predictors $ X = (X_1,$\\$\ldots,X_p) \in\mR^{n \times p} $ and error $ \e $:
\[y = X\beta + \e.\]
We are interested in estimating $ \beta $ with high-dimensional ($ n \ll p$) and sparsity setting, i.e. let $ q $ be the nonzero numbers of $ \beta $ and $ q \leqslant n $. Further, let $ X $ be the random samples of a multivariate normal distribution $\mN_p(0, \Sigma)$. We wish to estimate the precision matrix $ \Theta = (\theta_{jj'})_{j,j'=1,\ldots,p} = \Sigma^{-1} $.
The sparse $ \Theta $ corresponds to a Gaussian graphical model $ G=(V, S) $, where $ V = \{1, \ldots, p\} $ is the set of nodes corresponding to the random variables and $ S $ is a set of undirected edges representing the conditional dependence relationships between the variables. In both regression and graphical model, the dimension $ p $ is allowed to grow as $ n $ increases. For notational simplicity, we do not index the parameters with $ n $.

We propose the following two-step estimator for $ \beta $. We add two selection penalties to both functions respectively and construct a sparse estimator of regression coefficients depending on $ \Theta $.
\begin{center}
\textbf{Sparse Laplacian Shrinkage with the Graphical Lasso Estimator (SLS-GLE)}
\end{center}
Step 1 (Estimation of the precision matrix)
\begin{equation}\label{eq theta}
\mbox{Let}~ \hat \Theta = \arg \min\limits_{\Theta} \{-\log |\Theta| + \mbox{tr} (\Theta \hat \Sigma) + \lam_0 \|\Theta\|_1\},
\end{equation}
where $\hat{\Sigma}$ is the sample covariance matrix.

Step 2 (Estimation of the regression coefficients)
\begin{align}\label{eq beta}
\mbox{Set}~\hat \beta  = \arg \min\limits_{\beta} \big\{ \frac{1}{{2}}\|y - X\beta\|_2^2 + P_{\lam_1}(\beta) + \dfrac{\lam_2}{2}\sum\limits_{1 \leqslant j < j' \leqslant p} |\hat \theta_{jj'}|(\beta_j - s_{jj'}\beta _{j'})^2\big\},
\end{align}
where $ s_{jj'} = \text{sign}(\hat \theta_{jj'}) $. The first step estimates the precision matrix $ \Theta $ by applying the Glasso, an efficient estimate of the Gaussian graphical model, proposed by \citet{yuan2007model}. In the second step, $ P_{\lam_1}(\beta) $ is the penalty for selection and sparsity of regression features. In this paper, We apply the convex penalty $ P_{\lam_1}(\beta) = \lam_1\|\beta\|_1 $. There are several other selection penalties, such as the MCP \citep{zhang2010mcp} and the SCAD \citep{fan2001variable}; many of them have been shown to perform better over the $ l_1 $ penalty when features were independent. However, according to our simulation and empirical experience, the proposed method, using a graph penalty and the Lasso penalty, is more stable and often outperforms other penalties in dealing with complex datasets. To illustrate our point, we give a simple example in Figure~\ref{fig21} before readers get more detail comparisons in simulations. In Figure~\ref{fig21}, we compare the different solution paths of our method with the Lasso penalty and the MCP penalty according to the mean values with 100 repetitions. The former has better performance and a smaller estimation error. More details of the comparisons on different penalty functions can be found in the simulation part.

\begin{figure}[ht]
    \centering
    \subfigure[SLS-GLE with Lasso penalty]
    {\includegraphics[width=.48\textwidth,height=.47\columnwidth]{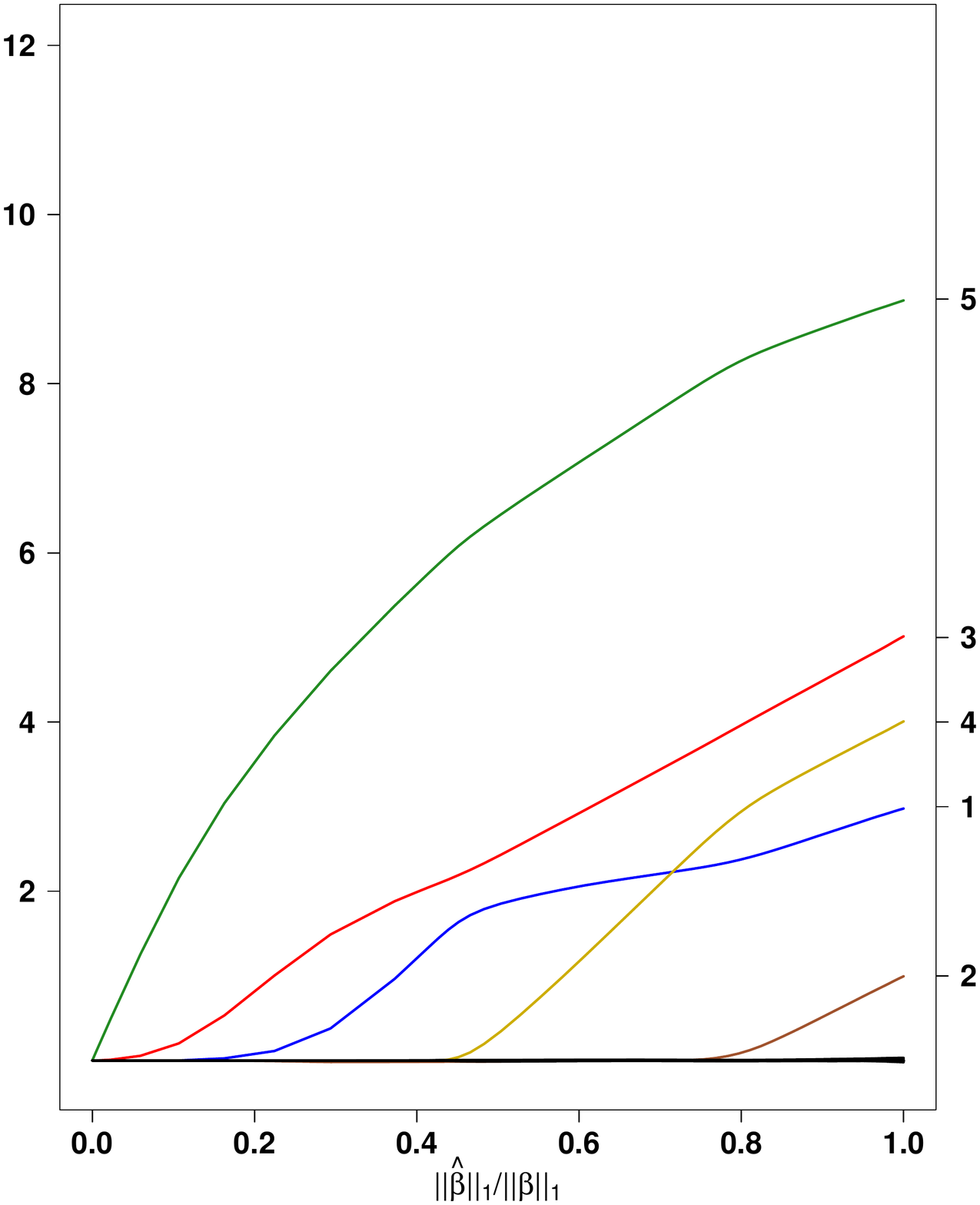}}
    \subfigure[SLS-GLE with MCP penalty]
    {\includegraphics[width=.48\textwidth,height=.47\columnwidth]{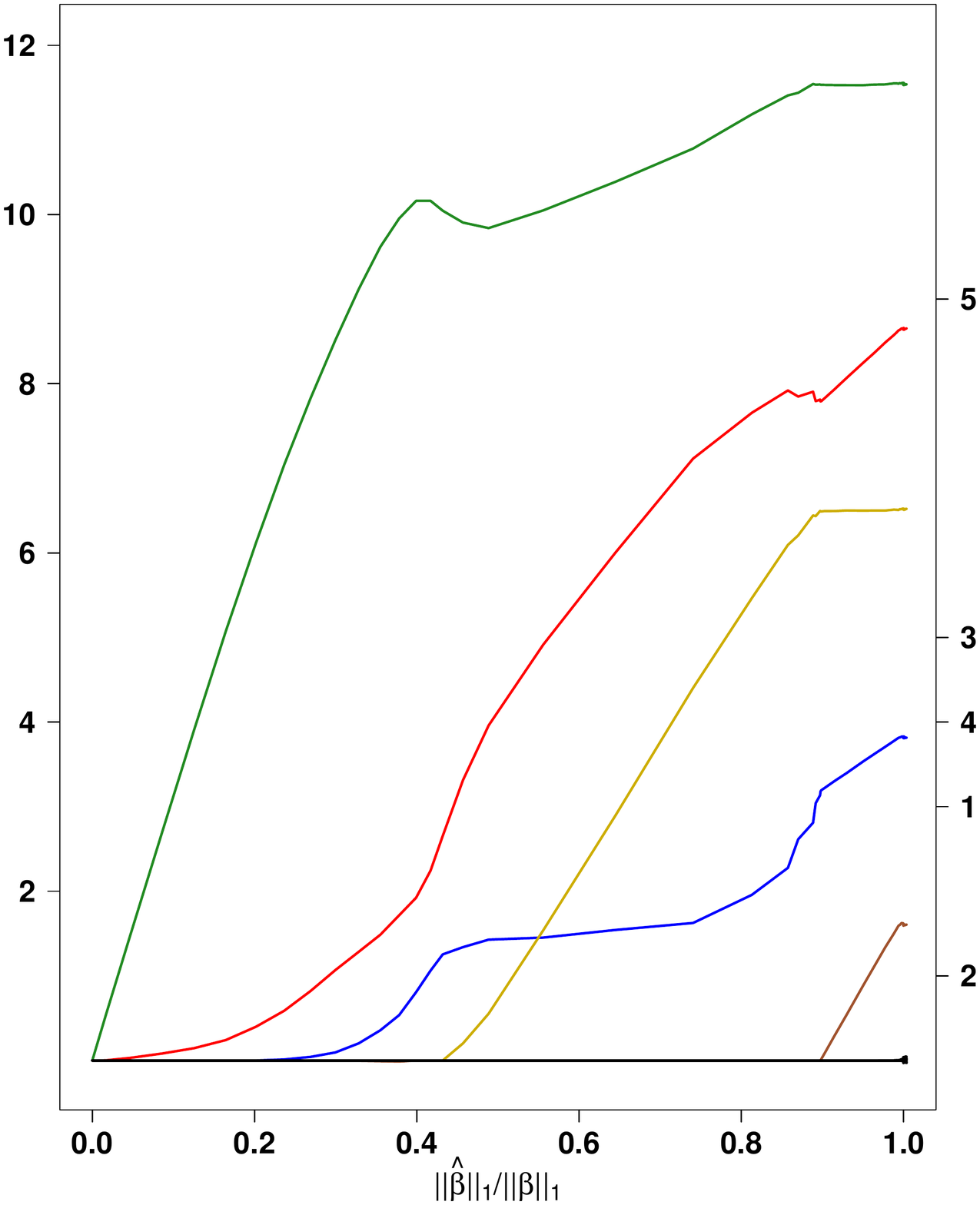}}
	\caption{Data: $n=50$, $p=20$, $\beta=(3,1,5,4,9,0,\dots,0)^\t$. $\e_i$ is generated from $\mN(0,1)$ and each row of $X$ is independently generated from $\mN_p(0,\Theta^{-1})$, where $\Theta$ is block diagonal with same 4 blocks, each of size $5\times 5$, and all the diagonal elements of blocks are 1 and the first off-diagonal elements are $0.5$.}
    \label{fig21}
\end{figure}

The second term of~\eqref{eq beta} is a Laplacian quadratic penalty \citep{chung1997spectral} and the assigned weights to the $ l_2 $ norm come from the adjacency matrix. In this paper we construct the adjacency matrix by the estimated precision matrix $ \hat \Theta $. The estimated Laplace matrix is denoted by $ \hat \Gamma = \hat D - \hat \Theta $ where  $ \hat D = \text{diag}(\hat d_1,\ldots,\hat d_p) $ and $ \hat d_j = \sum\limits_{j'=1}^p|\hat \theta_{jj'}| $. More details on Laplacian matrix are given in Appendix. In this way, \eqref{eq beta} can be written as
\begin{equation}\label{eq beta2}
\hat \beta  = \arg \min\limits_{\beta} \big\{ \frac{1}{{2}}\|y - X\beta\|_2^2 + \lam_1\|\beta\|_1 + \dfrac{\lam_2}{2}\beta^\t\hat \Gamma \beta\big\}.
\end{equation}
It is well known that the Laplacian matrix is a representation of a graph and incorporates the specific information on the complex pattern among the variables into regression. As an adaptive weight, the adjacency matrix is required to be sparse. Several forms of adjacency measures are studied, \citep{zhang2005general,huang2011laplacian,daye2009fusion}. The precision matrix represents the conditional dependence relationships between the variables, equivalently to the estimation of undirected Gaussian graphical model. We will provide evidence in the next Section and also simulations that the proposed method can yield robust and accurate results.

The proposed method enjoys the computational advantage of several efficient algorithms, which we can solve SLS-GLE as following.
\begin{prop}
Given data set $ (y, X) $, the estimation from the first step $ \hat \Theta $ and parameter $ \lam_2 $, define an artificial data set $ (y^*, X^*) $ by
\[ X^*_{(n+p) \times p} = (1+\lam_2)^{-1/2}
\begin{pmatrix} X\\\sqrt{\lam_2}\hat L \end{pmatrix}, \ \ \ y^*_{(n+p)} =
\begin{pmatrix}y \\0\end{pmatrix}, \]
where $ \hat L $ is a $ p \times p $ matrix that $ \hat L^\t \hat L=\hat \Gamma $. Then the SLS-GLE criterion can be written as
\[\hat \beta^*  := \arg \min\limits_{\beta}\bigg\{
\dfrac{1}{2n}\|y^* - X^*\beta\|^2_2  + P_{\lam_1}(\beta)
\bigg\},\]
then
\[ \hat \beta = \dfrac{1}{\sqrt{1+\lam_2}} \hat \beta^*. \]
\end{prop}

\section{Comparisons with other Methods}
To gain some insight into the performance of the SLS-GLE, in this section, we demonstrate several features of SLS-GLE and compare with several existing methods. We first state following theoretical support of the Glasso: According to \citet{ravikumar2011high}, with the incoherence assumption and proper choice of the tuning parameter, sample size, we have with high probability that
\[ \|\hat \Theta - \Theta\|_\infty \leqslant K \cdot (\dfrac{\log p}{n})^{1/2} \rightarrow 0 \ \text{as} \ n \rightarrow \infty.\]
In this section, following discussion will be conditional on the true graph estimated precisely by the first stage of SLS-GLE. More details between the true graph and the graph estimation from the first step will be discussed in the Appendix.

\subsection{Comparison with $ l_2 $ Penalty}
We begin by analysing the difference between our method and $ l_2 $-related methods, such as the Elastic Net \citep{zou2005elastic} and Mnet \citep{huang2016mnet}. We assume there exist several groups of variables among which the pairwise conditional dependencies are nonzero and all others are zero. For example, we assume $ X_j $ and $ X_{j'} $ are in one group, then $ \theta_{jk} $ and $ \theta_{j'k} $ are either both equal to zero or not. For simplicity the notational dependence of $ \theta_{j'k} $ and $ \theta_{jk} $, we assume that if $ \theta_{jj'} >0$, then $ \theta_{j'k} $ and $ \theta_{jk} $ can be replaced by each other with a negligible difference; similarly, if $ \theta_{jj'} <0 $, then $ -\theta_{j'k} $ and $ \theta_{jk} $ can be replaced by each other. Then we have the following result to describe the grouping property of the SLS-GLE, and briefly compare it with the Elastic Net.
\begin{prop}\label{prop 1}
Assume $ X_j $ and $ X_{j'} $ are conditional dependence. Set $ z = y - X\hat \beta $ and $z^\en = y - X \hat \beta^\en $ where $ \hat \beta $ is the solution of SLS-GLE, $ \hat \beta^\en $ is the solution of the elastic net. Then we have when $ \theta_{jj'} > 0 $,
\[|\hat \beta_j - \hat \beta_{j'}| = \dfrac{1}{\lam_2 (\sum_{k \neq j}|\theta_{jk}|+ \theta_{jj'})}\big|(X_j - X_{j'})^\t z\big|\]
and
\[|\hat \beta^\en_j - \hat \beta^\en_{j'}| = \dfrac{1}{\lam_2} \big|(X_j - X_{j'})^\t z^\en\big|,\]
when $ \theta_{jj'} <0 $,
\[|\hat \beta_j + \hat \beta_{j'}| = \dfrac{1}{\lam_2 (\sum_{k \neq j}|\theta_{jk}|- \theta_{jj'})}\big|(X_j + X_{j'})^\t z\big|\]
and
\[|\hat \beta^\en_j + \hat \beta^\en_{j'}| = \dfrac{1}{\lam_2} \big|(X_j + X_{j'})^\t z^\en\big|.\]
\end{prop}
In Proposition~\ref{prop 1}, $ \sum_{k \neq j}|\theta_{jk}| $ implies the influence of $ X_j $ affecting other variables in the graph.
According to our assumption, it can be seen as a measure of the important pattern of both $ X_j $ and $ X_{j'} $. $\theta_{jj'}$ implies the conditional dependence between $X_j$ and $X_{j'}$. When $ \theta_{jj'} > 0 $, the difference between two estimates, $ \hat \beta_j $ and $ \hat \beta_{j'} $, decreases as $ X_{j} $ (or $ X_{j'} $) becomes more important in the graph or the positive dependence between $ X_j $ with $X_{j'}$ becomes stronger, while at the same time, $\hat \beta_j + \hat \beta_{j'}$ decreases when $ \theta_{jj'} <0 $.

Comparing with the Elastic Net, we see another significant advantage of the SLS-GLE on grouping effect, that is, when a group is important to the response, by means one or more variables in the group identified as relevant, then the rest variables in this group are hard to be identified as irrelevant.
We can see how we obtain this result from a simple example, i.e., if we need to achieve $ \hat \beta_j =0 $, the following must hold
\[|X^\t_j(y - X\hat \beta)+\lam_2 \sum_{ k \neq j} \theta_{jj'}\hat \beta_k | \leqslant \lam_1.\]
As a contrast, we have $ \hat \beta^\en_j = 0 $ as long as the following holds
\[|X^\t_j(y - X\hat \beta)| \leqslant \lam_1.\]

\subsection{Comparison with other Adjacency Matrix}
There has been much work on modeling the data with correlated covariates, among, many of them consider using the Laplacian matrix associated with the covariates to promote the similarities among coefficients. The main difference between each method is the different constructions of the adjacency matrix. We review results from the literature, and give a detailed discussion to show that the precision matrix is more suitable for estimating the sparse regression with complex datasets. We will exhibit our points in the following.

We begin by recalling several forms of adjacency measures which are proposed and used in \citet{huang2011laplacian} and \citet{daye2009fusion}. We apply their notations and denote the adjacency matrix by $ A = (a_{jj'})_{jj'=1,\ldots,p} $. For the Laplacian quadratic penalty:
\[ \beta^\t \Gamma \beta = \sum\limits_{1 \leqslant j < j' \leqslant p} | a_{jj'}|(\beta_j - s_{jj'}\beta_{j'})^2.\]
There are five adjacency measures and all calculated by the Pearson correlation $ r_{jj'} $. Among, 1) and 2) are sparse adjacency functions with a  threshold function $ r $ determined based on the Fisher transformation: 1) $ a_{jj'} = \mathbf{1}\{r_{jj'}>r\} $, $ s_{jj'} = 1 $. 2)$ a_{jj'} = s_{jj'} \mathbf{1}\{|r_{jj'}|>r\} $, $ s_{jj'} = \sign{r_{jj'}} $. 3) and 5) are non-sparse adjacency functions: 3) $ a_{jj'}=\max(0, r_{jj'})^k $, $ s_{jj'} = 1 $; 4) $ a_{jj'} = |r_{jj'}|^k $, $ s_{jj'} = \sign{r_{jj'}} $; 5) $a_{jj'}=|r_{jj'}|^k(1-|r_{jj'}|)^{-1}$, $s_{jj'} = \sign{r_{jj'}} $ where $k$ is a positive constant. Comparing with above adjacency matrices, SLS-GLE has several important advantages:
\begin{itemize}
\item Many of adjacency measures do not consider the negative correlations. Our penalty pays attention to distinguish between negative and positive conditional dependencies. For example, reviewing Proposition~\ref{prop 1}, we have when $ \theta_{jj'} <0 $,
\[|\hat \beta_j + \hat \beta_{j'}| = \dfrac{1}{\lam_2 (\sum_{k \neq j}|\theta_{jk}|- \theta_{jj'})}\big|(X_j + X_{j'})^\t z\big|\]
Our penalty makes the difference between $ \hat \beta_j $ and $ \hat \beta_{j'} $ becomes larger when $ |\theta_{jj'}| $ increases, while many Laplacian shrinkage forms lead to the opposite direction.
\item In order to make the problem feasible, we always assume sparsity in the predictor space, i.e., 1) most of the predictors are unrelated to the response; 2) most of the predictors are unrelated to each other. Thus, we need an approach that can automatically perform high dimensional edge selection on the graph structure. In light of this, the Glasso is more suitable than the combination of the Pearson correlation and the threshold function.
\item Furthermore, the sparsity pattern of the estimated precision matrix can be easily tuned by the tuning parameter of Glasso. Its computation cost is smaller than the computation cost of finding a suitable threshold function in dealing with each correlation pattern.
\item Conditional dependence information often outperforms the correlation in constructing an adjacency matrix. For example, considering following dataset:
\[ X_i \leftrightarrow X_j \leftrightarrow X_k \leftrightarrow \cdots \leftrightarrow X_l, \]
where $ 1\leq i,j,k,\ldots,l\leq p $ and the arrow denotes the connection between two variables. If we use sample correlation to produce this predictor graph, above variables are all connected, increasing unnecessary computation burden and the bias on estimation.
\end{itemize}

\section{Theoretical Results}
In this section, we study the theoretical properties of the proposed estimator. We first introduce the following notations:
let $S=\{j:\beta_j \neq 0\}$ with cardinality $|S|=q$, similarly, $\hat S=\{j:\hat\beta_j \neq 0\}$. By construction of the subspaces, $\beta$ and $ X $ can be written in the partitioned form $\beta=(\beta_S^\t,\beta_{S^c}^\t)^\t$, $X = (X_{S},X_{S^c})$, respectively. The set of non-zero entries in the precision matrix is denoted by
\[S_\theta=\{ (j,j')\in\{ j\neq j',\theta_{jj'} \neq 0 \} \cup\{(1,1),\dots,(p,p)\}.\]
Suppose the variables are centered and normalized such that $ \dfrac{1}{n}\sum\limits_{i=1}^n x^2_{ij} =1$ for $j=1, \ldots, p$. Considering the dimensionality $p=O(e^{n^{c_3}})$ and $q=O(n^{c_1})$ where $0<c_1+c_3<1$, we have following assumptions:

\begin{assumption}\label{assum1}
Assume $X$ satisfy the restricted eigenvalue (RE) condition:
\begin{align}
\frac{{\left\| X\Delta \right\|_2^2}}{n} \geqslant K_1 \left\|\Delta \right\|_2^2, ~ \text{for all}~\Delta \in B,
\end{align}
where $K_1>0$ and $B= \{\Delta \in \mR^p:\left\| \Delta_{S^c} \right\|_1 \leqslant 3\left\| \Delta_{S} \right\|\} $.
\end{assumption}
\begin{assumption}\label{assum3}
There exists some $\alpha_2\in (0,1]$ such that
\begin{align}
\max_{e \in {S^c_\theta}} \left\|\rm M_{e}{( \rm M_{S_\theta})^{ - 1}}\right\|_1 \leqslant 1 - \alpha_2.
\end{align}
where $ \rm M_{S_\theta}$ and $\rm M_{e}$ denote the matrices with rows and columns of $\rm M$ indexed by $S_\theta \times S_\theta$ and $ e \times S_\theta $ where $ e \in S^c_\theta $, respectively.
\end{assumption}

Assumption~\ref{assum1} requires a restriction of the generalized Gram matrix to the columns in $ S $ is invertible, which is widely used to bound the $ l_2 $-error between $ \beta $ and the estimate \citep{bickel2009simultaneous,meinshausen2009lasso}.
Assumption~\ref{assum3} is defined in \citet{ravikumar2011high},
which shares an exact parallel with the irrepresentable condition for the Lasso, and limits the influence that the non-edge terms, indexed by $ S^c_\theta $, can have on the edge-based terms, indexed by $ S_\theta $. Above two assumptions are used for SLS-GLE achieving the upper bound of  $ l_2 $-norm error with high probability.

\begin{assumption}\label{assum2}
    Let $C=\dfrac{1}{n} X^\t X$, $C_{S}=\dfrac{1}{n}X_S^\t X_S$, $C_{S^c}=\dfrac{1}{n}X_{S^c}^\t X_S$. There exists a constant $\alpha_1 \in (0,1]$ such that
    \begin{align}
        \|\big(C_{S^c}+\frac{\lam_2}{n} \Gamma_{S^c}\big)\big(C_S+\frac{\lam_2}{n}\Gamma_S \big)^{-1}\big[sign(\beta_S)+\frac{\lam_2}{\lam_1}\Gamma_S \beta_S\big] -\frac{\lam_2}{\lam_1}\Gamma_{S^c} \beta_S \|_\infty<1-\alpha_1.
    \end{align}
    where $\Gamma_S$ and $\Gamma_{S^c}$ denote the matrices with rows and columns of $\Gamma$ indexed by $S \times S$ and $ S^c \times S $, respectively.
\end{assumption}
\begin{assumption}\label{assum4}
    There exist positive constants $c_2$ and $K_2$ such that $c_1+c_3<c_2 < 1$ and
    \begin{align*}
        n^{\frac{1-c_2}{2}}\min_{i\in S}|\beta_i|\geqslant K_2.
    \end{align*}
\end{assumption}
Assumption~\ref{assum2} can be seen as the irrepresentable condition when $\lam_2 =0$. This condition is commonly used in the $ l_1 $ penalized regularizations to recover the true model with probability tending to one, see \citet{zhao2006lasso,wainwright2009sharp,li2010graph}. Assumption~\ref{assum4} requires a small gap between $ \beta_S $ and $ 0 $ and allows the nonzero coefficients tend to zero when $ n \rightarrow \infty$. These two assumptions are used for SLS-GLE obtaining the variable selection consistency. All the Assumptions have appeared frequently in the literature for proving theoretical properties of penalized regularizations. Now we state our major theoretical results.
\begin{thm}\label{thm 1}
Suppose Assumption~\ref{assum1}~-~\ref{assum3} hold. If $\lam_2 \|\beta\|_2\leqslant \lam_1\sqrt{q}$ and $\lam_1=4\sigma\sqrt{n\log{p}}$. Set $K_3=(2+\lam_{\min}(\Gamma))^{-1}K_1$. The following event holds with probability at least $ 1-1/p$:
\[\| \hat\beta-\beta \|_2 \leqslant \dfrac{8\sigma}{K_3} \sqrt{\dfrac{q \log p}{n}}.\]
\end{thm}

\begin{thm}\label{thm 2}
Suppose Assumption~\ref{assum3}~-~\ref{assum4} hold. If $\lam_2 \|\beta\|_2\leqslant \lam_1 \sqrt{q}$ and $\lam_1=4\sigma\sqrt{n\log{p}}$. We have
\begin{align*}
P(\hat S=S)\geqslant 1-1/p\rightarrow 1, \ as \ n\rightarrow\infty.
\end{align*}
\end{thm}

\section{Simulations and Application}
\subsection{Computations}
To compute the proposed method, we adopt the coordinate wise descent algorithm \citep{friedman2007pathwise,friedman2010regularization}.
This algorithm is originally proposed for the Lasso, and has been widely applied to calculate many estimates, such as the Fused Lasso, LAD-Lasso, Elastic Net, MCP, etc. It is very competitive with other procedures in the high-dimensional setting for the simple calculations and fast iterations, and also apply to our method.

The computation details are given in Algorithm 1. As mentioned above, the optimization problem \eqref{eq beta} is equal to the optimization problem \eqref{eq beta2} and we use $\hat\tau_{jj'}$ as the element of $\hat\Gamma$ where $ j,j'=1,...,p $.
\begin{center}
\textbf{Algorithm 1}
\end{center}
\begin{description}
    \item ~~Given $\lambda_1$, $\lambda_2$ and an initial value of $\hat\beta^{(0)}$.
    \item ~~Step 1: Solve the first step \eqref{eq theta} by the Glasso and calculate the Laplacian matrix $\hat\Gamma$ by the graph estimation $ \hat \Theta $.
    \item ~~Step 2: Set $\hat \beta^{(m)} \gets \hat\beta^{(m-1)}$. Visit all entries of $\hat\beta^{(m)}$. For each entry $j$, update $\hat\beta^{(m)}_j$ with the minimizer of the objective function along its coordinate direction given by
    \[\hat \beta_j^{(m)} \gets \dfrac{S(\sum\limits_{i=1}^n x_{ij}(y_i-\widetilde{y}_i^{(m)})-\lam_2 \sum\limits_{j'\neq j} \hat\tau_{jj'}\hat\beta_{j'}^{(m)},\lam_1 )}{1+\lam_2 \hat\tau_{jj}}, \]
    where $\widetilde{y}_i^{(m)}=\sum\limits_{j'\neq j} x_{ij'}\hat\beta_{j'}^{(m)}$ and the function $S(a,b)=\sign{a}(|a|-b)_+$.
    \item ~~Step 3: if $\sum\limits_j |\hat\beta_j^{(m)}-\hat\beta_j^{(m-1)}|< \e$ then stop, otherwise go to step 2.
\end{description}

Algorithm 1 is guaranteed to converge to the global minimizer as the give $\hat \Gamma$ in step 1 is nonnegative definite. We set the convergence tolerance parameter $\e=10^{-4}$ and $\hat\beta^{(0)}=\frac{X^\t y}{n}$.

\subsection{Some Numerical Experiments}
In this part, we conduct simulations to illustrate the performance of the proposed method, comparing with other methods: 1-2) SLS with two adjacency functions \citep{huang2011laplacian}; 3) Penalized regularization based on estimated covariance matrix (PRECM) without total variation penalty\citep{li2018graph}; 4) Weighted fusion \citep{daye2009fusion}, 5) Elastic Net and 6) Lasso.
We use four examples to state the performance of each method. During all the examples, we fix $ p = 300 $ while $ n=\{50,60,70,80,90,100\} $ so that $ p \gg n $. The elements of $ \beta $ are either equal to 0 or 1.5, and we set $ q = 10 $ to control the sparsity of $ \beta $. Here are the details of four scenarios.
\begin{description}
    \item Example 1. $ \Sigma $ is block diagonal with same 60 blocks, each of size $ 5 \times 5 $. All the diagonal elements of blocks are 1 and off-diagonal elements are equal to $ 0.8^{|j - j'|} $, where $ j,j'=1,...,5 $.
	\item Example 2. Set the correlations as $ \rho_{jj'} =0.8^{|j - j'|} $ for $ j,j'=1,...,p $.
	\item Example 3. $ \Theta $ is block diagonal with same 60 blocks, which are $ 5 \times 5 $ matrices. All the diagonal elements of blocks are 1 and the first off-diagonal elements are 0.5, other elements of blocks are equal to zero, i.e. Let $ B = \{b_{jj'}\}_{jj'=1,...,5} $ denotes the block and we have
	\[b_{jj'} = \left\{\begin{tabular}{ll}
	1,& if $ j=j' $\\
	0.5,& if $ |j-j'|=1 $\\
	0,& others.
	\end{tabular}\right.\]
	\item Example 4. Set $\Theta = F + \delta I$, where the diagonal entries of $F$ are all equal to zero, and each off-diagonal entry is generated independently and equals 0.5 with probability $0.02$ and $0$ with probability $0.98$. $ \delta $ is chosen such that the conditional number is equal to $ p $, and the matrix is then standardized to have unit diagonals.
\end{description}
In Example 1 and Example 2, we consider the different structures of the covariance matrix. For the other two cases, we consider the different structures of the precision matrix.
The tuning parameters are selected using BIC and all simulations are repeated 100 times. We use the $ l_2 $-norm and MSE to evaluate the performance of the estimate changes with the number of observations. The results are shown in Figure~\ref{fig1} and Figure~\ref{fig2}, while out-of-range values are deleted from the figures. Among, SLS-N2 and SLS-N4 are short for the SLS with adjacency function measured by N.2. and N.4. in \citet{huang2011laplacian} respectively.

As we can see from Figure~\ref{fig1} and Figure~\ref{fig2}, SLS-GLE performs the best among seven methods, especially in Example 1 and Example 3. Besides, it is easy to see that both $l_2$-norm and MSE of each method usually decrease as $n$ increase, which is obvious since as $n$ increases, we can get more information from data.

\subsection{Application to Real Data}
In this part, we apply our estimation scheme to the real dataset from the financial market. We aim to modeling the relationship between the index S$\&$P500 and its constituent stocks.
Financial market data structure is complex, where assets move in relation to each other, and random market fluctuations make it difficult to forecast the market's direction. To describe the market, investors and managers use the stock index, which is a measurement of the stock market computing from the selected stocks. Among, S$\&$P500 is a common index and benchmark for the American stock market, and is consisted of 500 largest companies. Since the sample size is often less than one hundred due to the time availability, this is a typical high-dimensional problem and it requires a sparse solution for the cost concern.

The data come from TXDB, containing the prices of stocks in S$\&$P500 from Jan. 2014 to Oct. 2018. The data is divided by time window, 5 months' data ($n=100$) for modeling and one month's data ($n=20$) for forecasting, which produces 52 forecasting sample sets. Let $x_{i,t}$ represents the price of $i$th constituent stock, $i=1,\ldots,500$ and $y_t$ represents the price of the index. We describe the relationship between $x_{i,t}$ and $y_t$ by a linear regression model.
We use the Annual Tracking Error (ATE) to be the measurement, which is a deviation measure of the return for the replication from the index, that is: $ \textmd{Tracking Error}_{\textmd{Year}} =\sqrt{252} \sqrt {\frac{\sum (\textmd{error}_i - \textmd{mean(error)})^2}{T - 1}} $, where $ \textmd{error}_i=r_i-\hat r_i $ and $r_t=\frac{y_t-y_{t-1}}{y_{t-1}},~t=1,\ldots,T$

We present the forecasting result in Table~\ref{table1} and Figure~\ref{fig3} - Figure~\ref{fig4}.
Table~\ref{table1} shows the predicted results of different methods with different subset sizes (20,40,60,80) in recent segments (Jan. 2018 - Aug. 2018 for modeling and Sept. 2018 for forecasting). As shown in Table~\ref{table1}, the proposed method nearly outperforms the other methods in forecasting the index, and its predicted ATE are nearly between 2\% and 5\%, which are qualified as an index fund in the market. Several mutual index fund achieves this standard by using all the constituent stocks.
As shown in Figure~\ref{fig3} and Figure~\ref{fig4}, we show the performance of SLS-GLE forecasting the index by selecting 50 and 70 constituent stocks during Jun. 2014 to Sept. 2018. The first row of both figures shows the predicted ATE of first 14 segments (Jun. 2014 - Feb. 2015). We can see that SLS-GLE forecasts the target index well and obtains smaller errors when we increase the selected set.

\section{Conclusions and Discussion}
We have proposed a procedure named SLS-GLE for correlated features in sparse high dimensional regression models. SLS-GLE uses the Gaussian graphical model to estimate the specific information on the conditional dependence pattern among the predictors. By combining the strengths of the Glasso and the Laplacian matrix penalty, the SLS-GLE procedure is intuitive, theoretically justified and easy to implement.

Our results offer new insights into the $ l_1 $-related method encouraging the grouping effect, and support a novel strategy of combining the regression model and the graphical model to improve the accuracy of complex datasets. For both theoretical analysis and numerical comparisons, we have carefully discussed the features of the SLS-GLE. It worth emphasizing that the SLS-GLE encourages sparsity on both regression model and graphical model, and aims to improve the estimation and prediction of features which are connected with each other.

There are several other open questions that we leave for future research. For example, we did not provide any discussion on the errors' distribution. Further, complex datasets not only denotes the data type with complex relationships between variables, but also includes heavy-tailed noise, influential observations, etc. For such complex data, fitting the model based on a clean data assumption may leads to a completely wrong solution. An effective deletion method or procedure like SLS-GLE, which achieves good estimation by measuring the specific information of the complexity, become more and more important in statistical analysis and many applications.

\section*{Acknowledgement}
This work was supported by the National Natural Science Foundation of China (Grant No. 11671059).
\clearpage

\begin{figure}[ht]
\centering
\subfigure[Example 1]
{\includegraphics[width=.48\textwidth,height=.5\columnwidth]{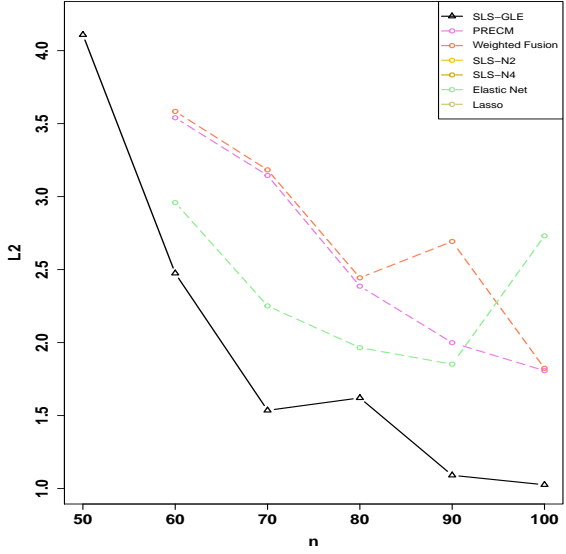}}
\subfigure[Example 2]
{\includegraphics[width=.48\textwidth,height=.5\columnwidth]{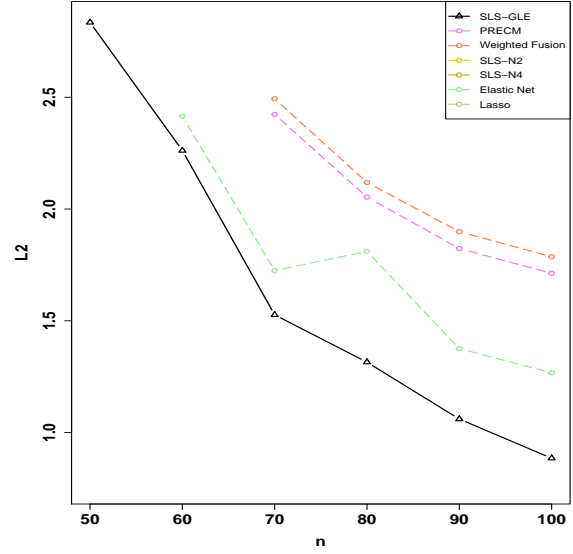}}
\subfigure[Example 3]
{\includegraphics[width=.48\textwidth,height=.5\columnwidth]{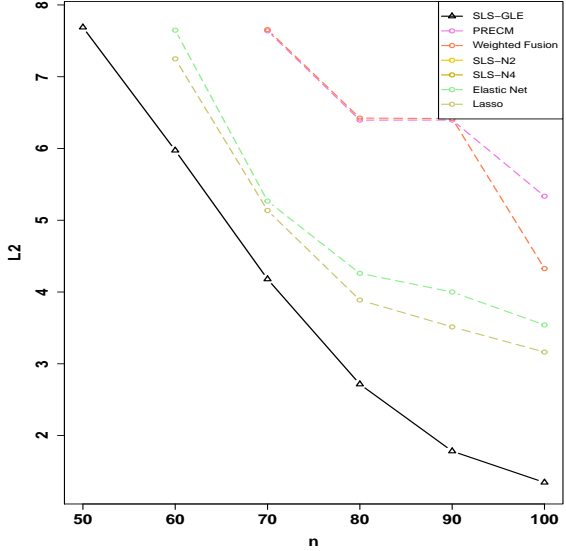}}
\subfigure[Example 4]
{\includegraphics[width=.48\textwidth,height=.5\columnwidth]{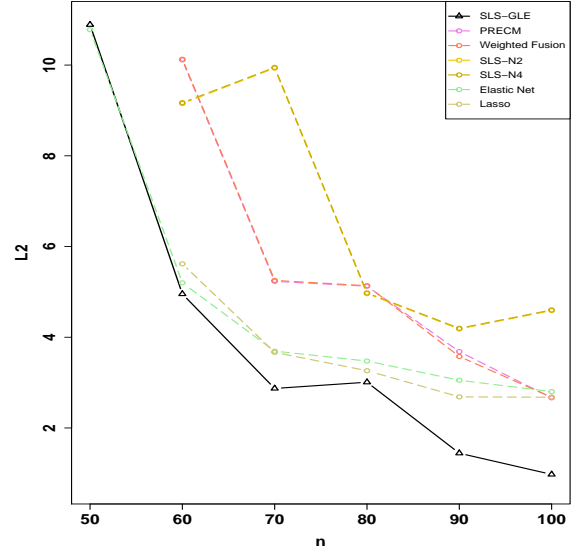}}
\caption{Results for four Examples. The black line shows the results for SLS-GLE and the lines with other colors correspond to the other methods: SLS-N2, SLS-N4, PRECM, Weighted Fusion, Elastic Net and Lasso. The vertical axis is the $l_2$ norm error and the horizontal axis corresponds to the number of samples ($n=50,60,70,80,90,100)$. }
\label{fig1}
\end{figure}
\clearpage

\begin{figure}[ht]
\centering
\subfigure[Example 1]
{\includegraphics[width=.48\textwidth,height=.5\columnwidth]{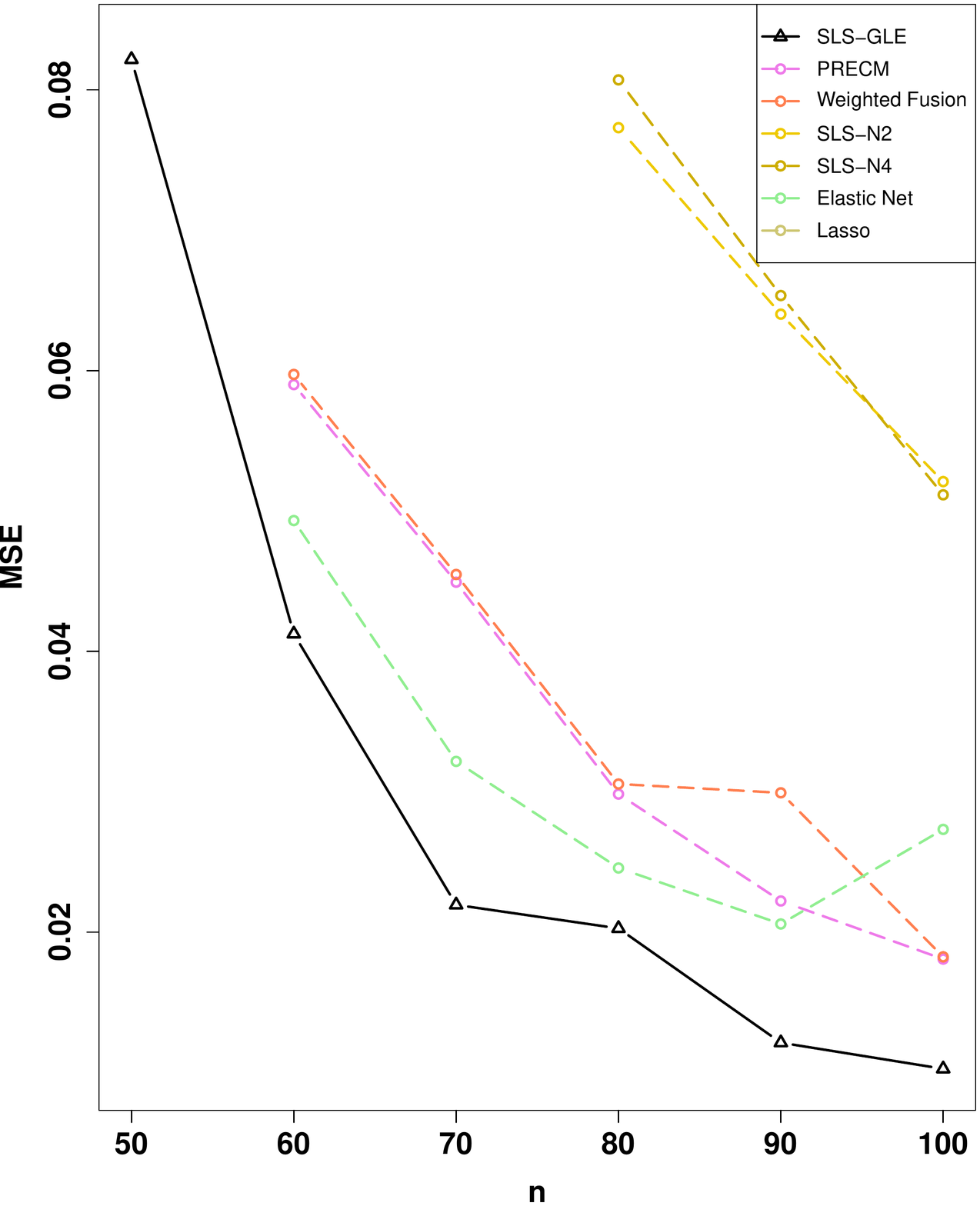}}
\subfigure[Example 2]
{\includegraphics[width=.48\textwidth,height=.5\columnwidth]{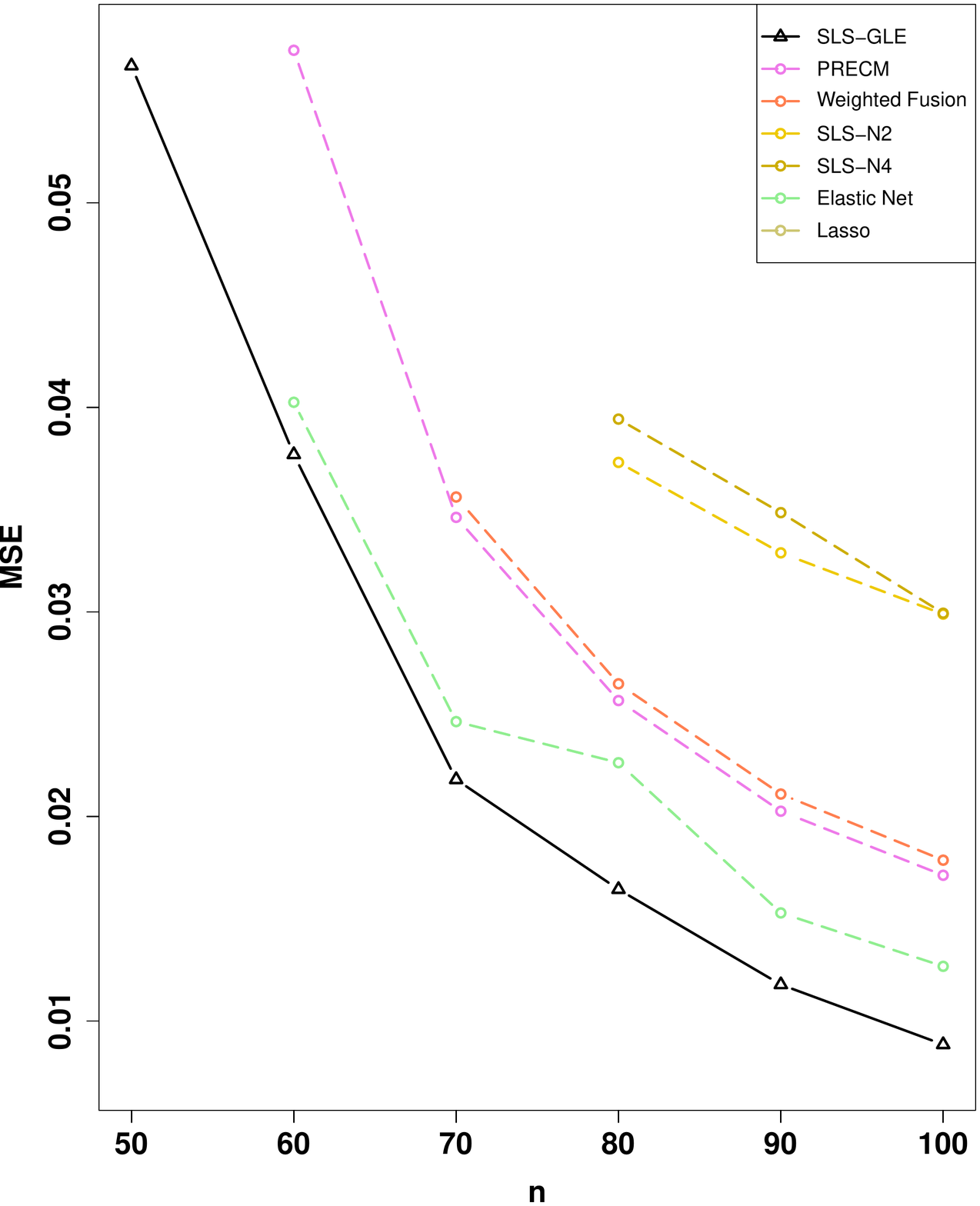}}
\subfigure[Example 3]
{\includegraphics[width=.48\textwidth,height=.5\columnwidth]{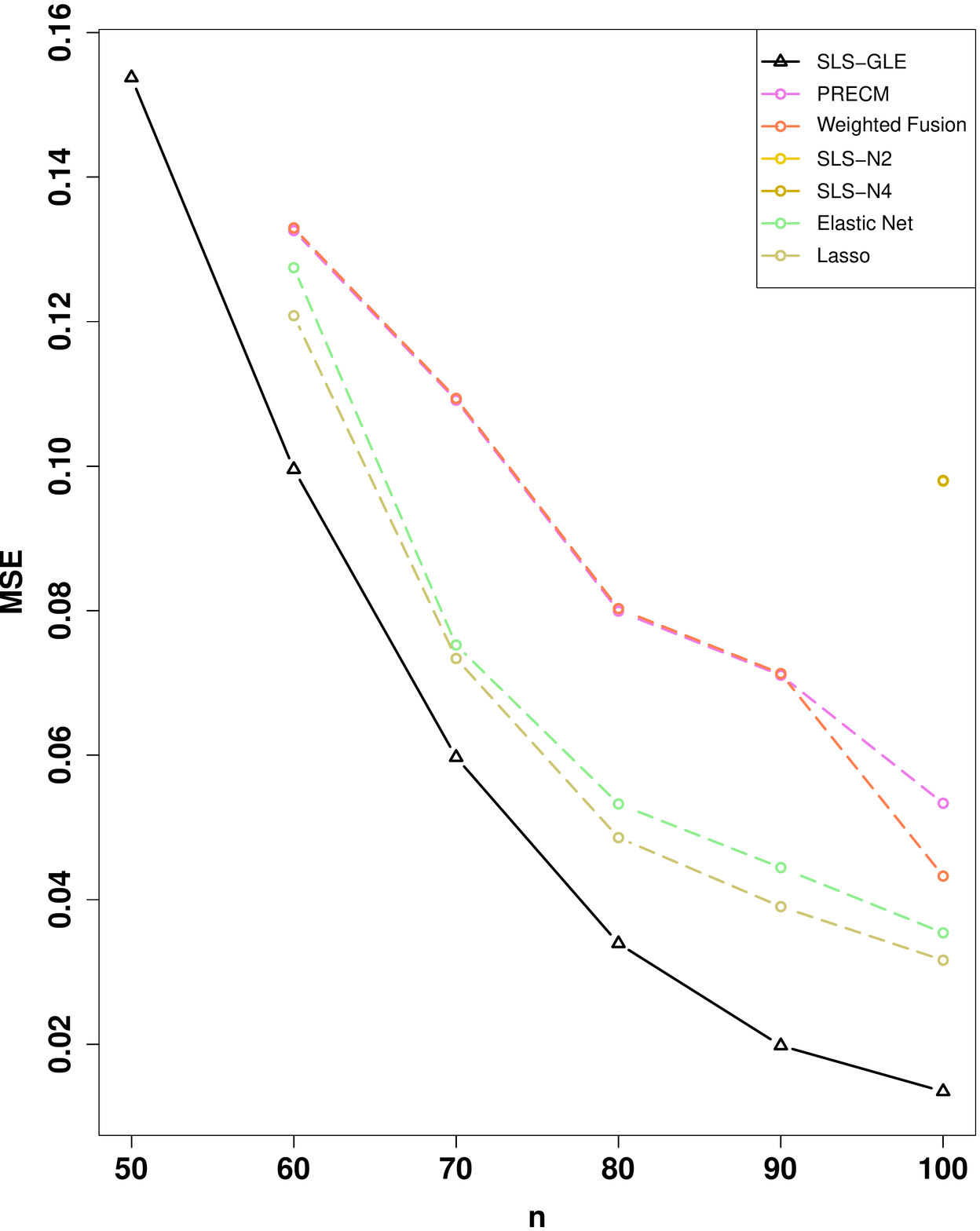}}
\subfigure[Example 4]
{\includegraphics[width=.48\textwidth,height=.5\columnwidth]{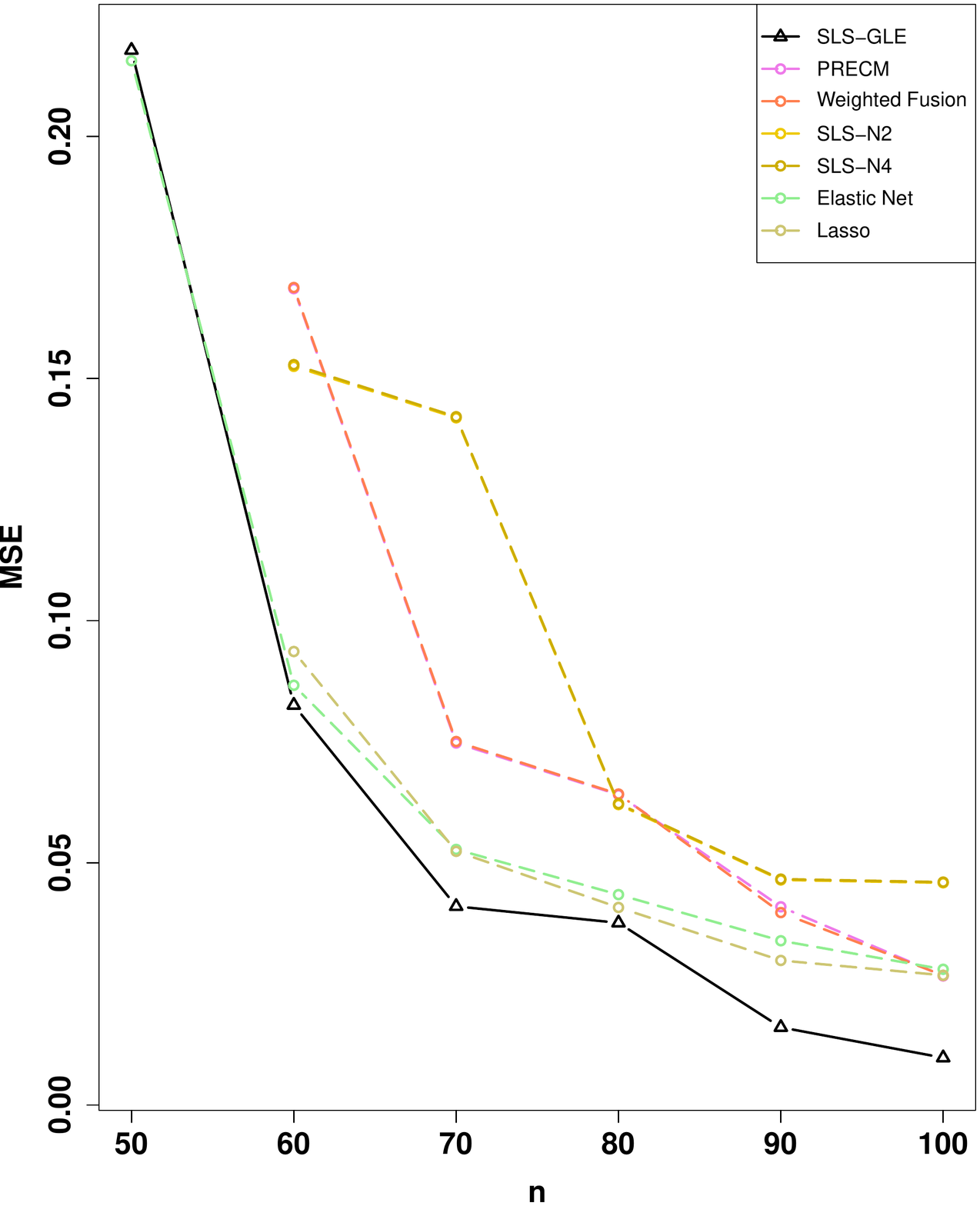}}\caption{Results for four Examples. The black line shows the results for SLS-GLE and the lines with other colors correspond to the other methods: SLS-N2, SLS-N4, PRECM, Weighted Fusion, Elastic Net and Lasso. The vertical axis is the MSE error and the horizontal axis corresponds to the number of samples ($n=50,60,70,80,90,100)$. }
\label{fig2}
\end{figure}
\clearpage

\begin{table}[ht]
	\centering
	\small
	\renewcommand\arraystretch{1.5}
	\setlength\tabcolsep{1pt}
	\caption{The predicted Annual Tracking Error(ATE).}
	\label{table1}
	\begin{tabular}{cccccc}
		\hline
		Selected stocks&Method& ATE&Selected stocks&Method& ATE\\
		\hline
		20& SLS-GLE & 4.39\% & 40 & SLS-GLE & 2.93\%\\
	      & PRECM & 4.31\% &    & PRECM & 2.83\%\\
	      & Weighted Fusion & 4.31\% &    & Weighted Fusion & 2.79\%\\
	      & SLS-N2 & 7.14\% &    & SLS-N2 & 6.28\%\\
	      & SLS-N4 & 7.14\% &    & SLS-N4 & 6.29\%\\
	      & Elastic Net & 10.48\% &    & Elastic Net & 9.7\%\\
	      & Lasso & 4.84\% &    & Lasso & 3.36\%\\
	      \hline
		Selected stocks&Method& ATE&Selected stocks&Method& ATE\\
		\hline
		60& SLS-GLE & 2.58\% & 80 & SLS-GLE &2.19\%\\
		  & PRECM & 2.96\% &    & PRECM & 2.6\%\\
		  & Weighted Fusion & 2.96\% &    & Weighted Fusion & 2.62\%\\
		  & SLS-N2 & 6.78\% &    & SLS-N2 & 9.06\%\\
		  & SLS-N4 & 6.78\% &    & SLS-N4 & 9.06\%\\
		  & Elastic Net & 8.42\% &    & Elastic Net & 7.05\%\\
		  & Lasso & 3.11\% &    & Lasso & 2.65\%\\
	    \hline
	\end{tabular}
\end{table}
\clearpage

\begin{figure}[ht]
	\centering
	\subfigure[The predicted ATEs of different methods]
	{\includegraphics[width=\textwidth,height=.35\columnwidth]{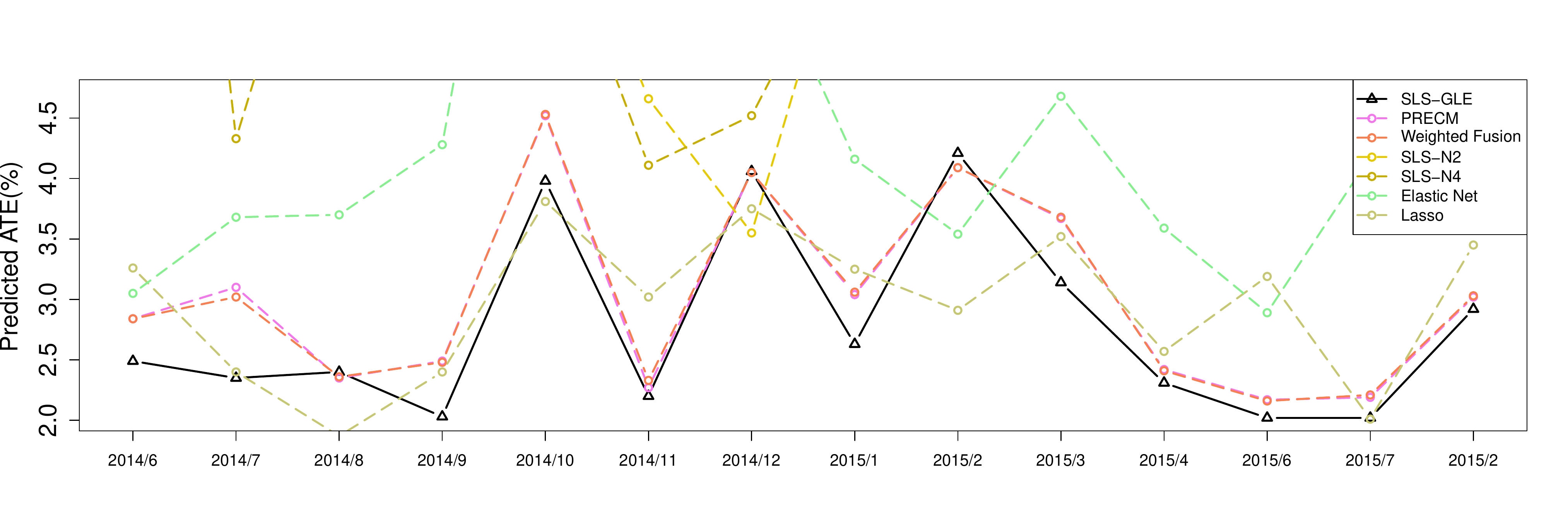}}
	\subfigure[Index tracking of SLS-GLE]
	{\includegraphics[width=\textwidth,height=.35\columnwidth]{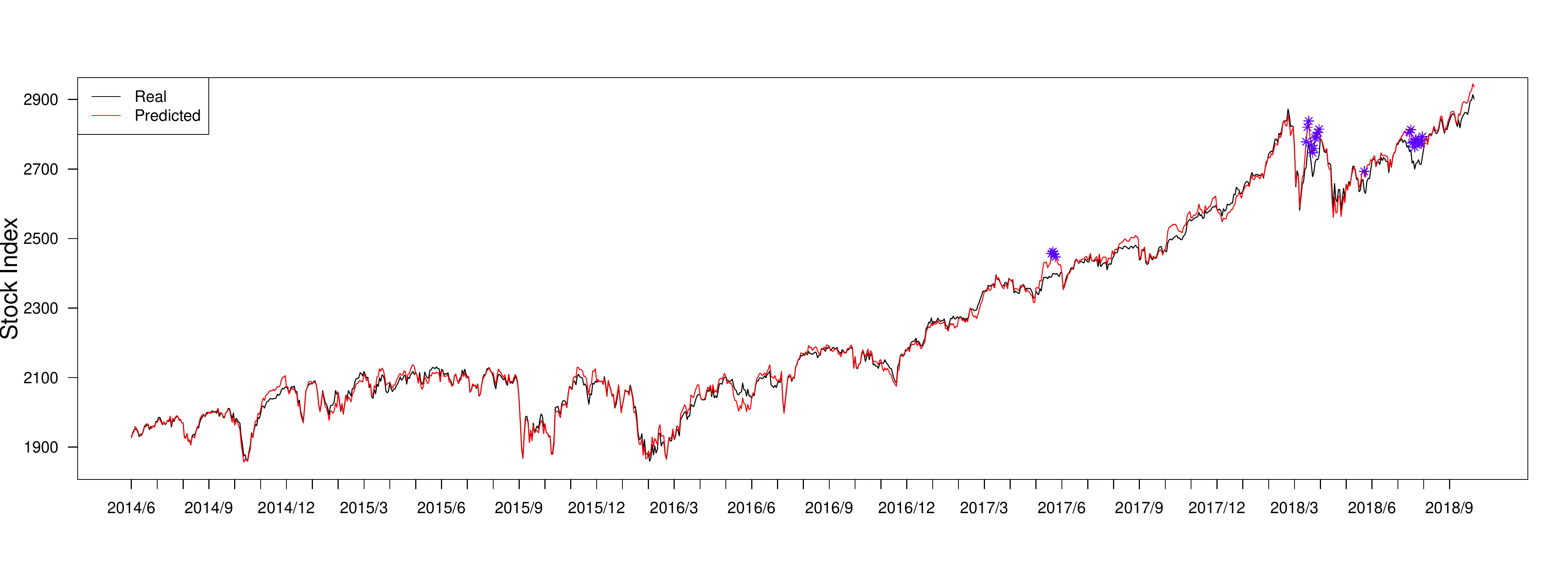}}
	\caption{The performance of selecting 50 stocks in index tracking. The first row shows the predicted ATE of SLS-GLE. The second row shows the forecasting index tracking result. Among, the blue dots denote the location where the absolute error between the predicted index and the real index exceeds 50.}
	\label{fig3}
\end{figure}
\clearpage

\begin{figure}[ht]
	\centering
	\subfigure[The predicted ATEs of different methods]
	{\includegraphics[width=\textwidth,height=.35\columnwidth]{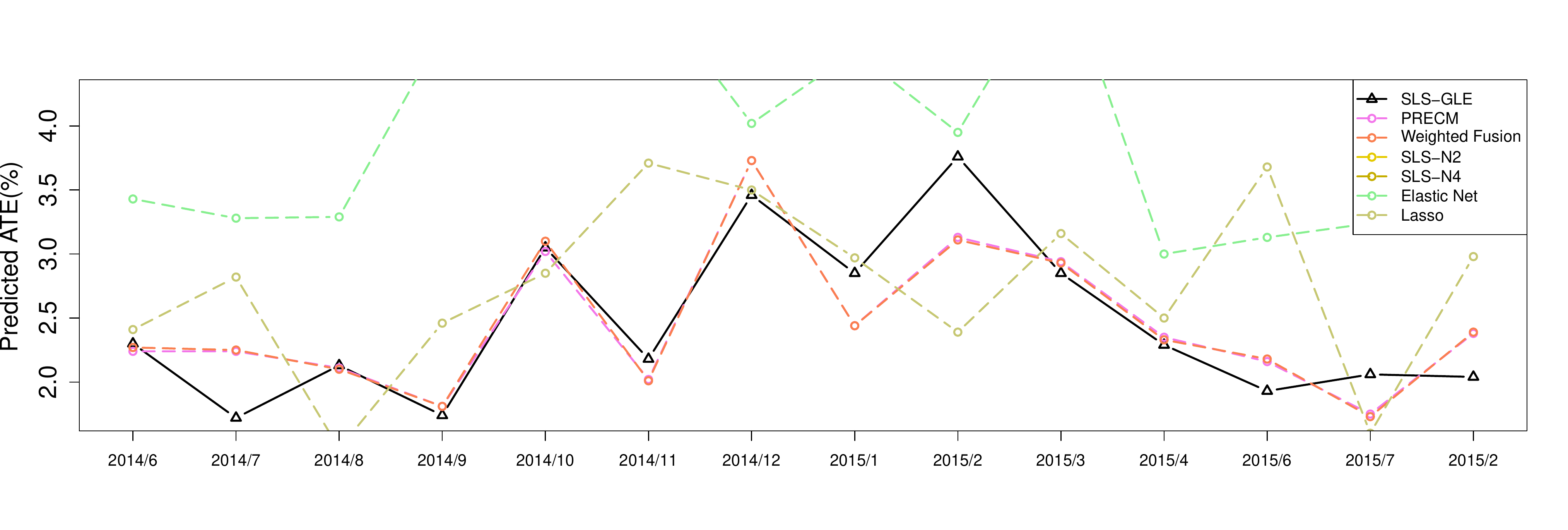}}
	\subfigure[Index tracking of SLS-GLE]
	{\includegraphics[width=\textwidth,height=.35\columnwidth]{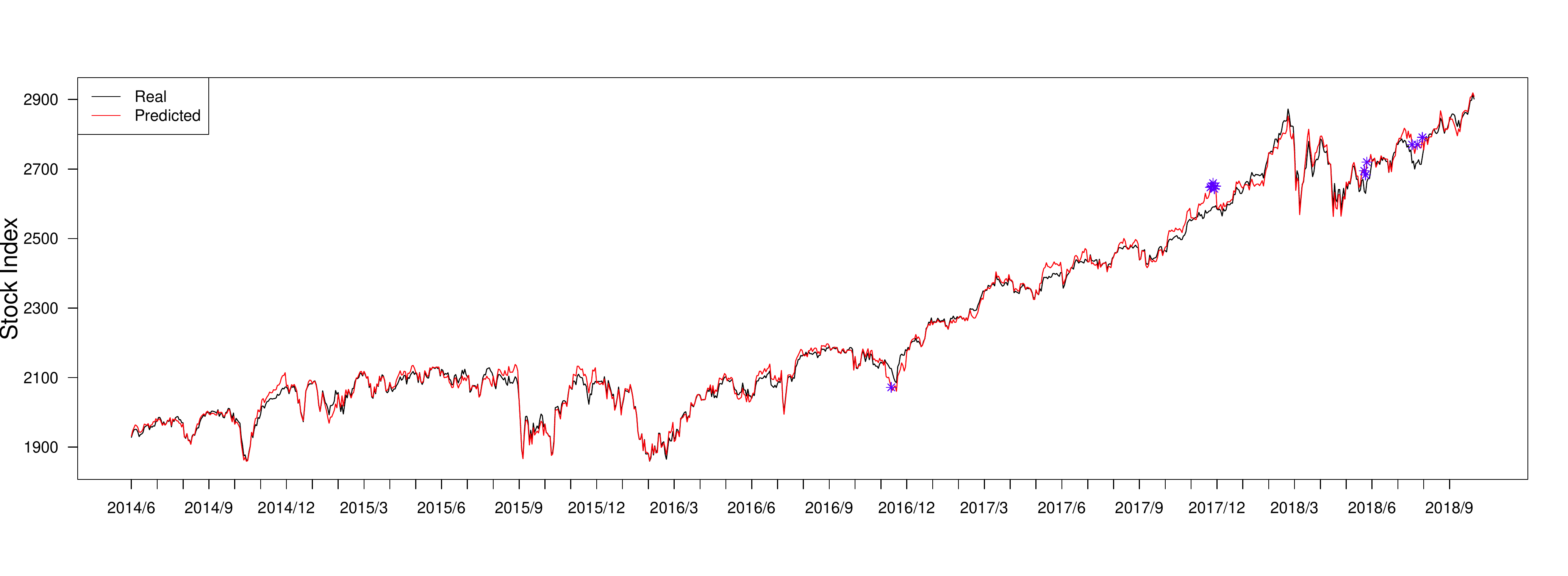}}
	\caption{The performance of selecting 70 stocks in index tracking. The first row shows the predicted ATE of SLS-GLE. The second row shows the forecasting index tracking result. Among, the blue dots denote the location where the absolute error between the predicted index and the real index exceeds 50.}
	\label{fig4}
\end{figure}
\clearpage

\appendix
\section*{Appendix}

\noindent \textbf{Laplacian Matrix}, proposed by \citep{chung1997spectral}, is a symmetric matrix representation of a graph. Given a graph, its Laplacian matrix $ \Gamma $ is defined as
\[\Gamma = D - A,\]
where $ A=(a_{j,j'}) $ is the adjacency matrix and $ D $ is the degree matrix with off-diagonal elements equal to zero and diagonal elements are $ d_j = \sum\limits^p_{j'=1}|a_{jj'}| $. Then we have
\[\beta^\t \Gamma \beta =\beta^\t (D-A)\beta= \sum\limits_{1 \leqslant j < j' \leqslant p} |a_{jj'}|(\beta_j - s_{jj'}\beta_{j'})^2,\]
where $ s_{jj'} = \text{sign}(a_{jj'}) $.

\begin{proof}[Proof of Proposition 2]
Given the tuning parameters, we have following form
\begin{equation}\label{eq 1}
-X^{\t}_j (y - X \hat \beta) + \lambda_1 \sign{\hat \beta_j} + \lambda_2\hat \beta_j \cdot \sum_{ k \neq j}|\theta_{jk}| - \lambda_2\sum_{ k \neq j} \theta_{jk}\hat \beta_k=0.
\end{equation}
Set $ z = y - X \hat \beta $. From~\eqref{eq 1}, we have
\begin{equation}\label{eq 11}
\hat \beta_j = \big(\lam_2 \sum_{ k \neq j}|\theta_{jk}|\big)^{-1} \big( X^\t_j z + \lam_2 \sum_{ k \neq j} \theta_{jk}\hat \beta_k - \lam_1 \sign{\hat \beta_j} \big).
\end{equation}
Assume that there exists two variables, $ j $ and $ j' $, in a group that are conditional dependence. When $ \theta_{jj'}>0 $, from~\eqref{eq 11}, we have
\[\hat \beta_j - \hat \beta_{j'} = H_1 +H_2 +H_3,\]
where
\begin{center}
\renewcommand\arraystretch{2}
\begin{tabular}{l}
$ H_1 = \dfrac{X^\t_j z}{\lam_2 \sum_{ k \neq j}|\theta_{jk}|}- \dfrac{X^\t_{j'} z}{\lam_2 \sum_{ k \neq j'}|\theta_{j'k}|},$\\
$H_2 = \dfrac{\sum_{ k \neq j}\theta_{jk}\hat \beta_k}{\sum_{ k \neq j}|\theta_{jk}|} -\dfrac{\sum_{ k \neq j'}\theta_{j'k}\hat \beta_k}{\sum_{ k \neq j'}|\theta_{j'k}|},$\\
$ H_3 = \lam_1 \cdot \lam_2 \sum_{ k \neq j}|\theta_{jk}|\cdot \sign{\hat \beta_j} - \lam_1 \cdot \lam_2 \sum_{ k \neq j'}|\theta_{j'k}|\cdot \sign{\hat \beta_{j'}}.$\\
\end{tabular}
\end{center}

According to the assumption that $ \theta_{j'k} $ and $ \theta_{jk} $ can be replaced by each other with a negligible difference, then we have
\[ H_1 = \dfrac{1}{\lam_2 \sum_{ k \neq j}|\theta_{jk}|}(X^\t_j z-X^\t_{j'} z)\]
and
\[H_2 = -\dfrac{\theta_{jj'}}{\sum_{ k \neq j}|\theta_{jk}|} \ (\hat \beta_j - \hat \beta_{j'}). \]
For $ H_3 $, we have $ H_3 = 0 $ when $ \sign{\hat \beta_j} = \sign{\hat \beta_{j'}} $. Combining $ H_1 $, $ H_2 $ and $ H_3 $, the difference between $ \hat \beta_j $ and $ \hat \beta_{j'} $ becomes
\[|\hat \beta_j - \hat \beta_{j'}| = \dfrac{1}{\lam_2 (\sum_{ k \neq j}|\theta_{jk}|+ \theta_{jj'})}\big|(X_j - X_{j'})^\t z\big|.\]
Similarly, when $ \theta_{jj'}<0 $, we have
\[\hat \beta_j + \hat \beta_{j'} = H^*_1 +H^*_2 +H^*_3,\]
\begin{center}
\renewcommand\arraystretch{2}
\begin{tabular}{l}
$ H^*_1 = \dfrac{X^\t_j z}{\lam_2 \sum_{ k \neq j}|\theta_{jk}|} + \dfrac{X^\t_{j'} z}{\lam_2 \sum_{ k \neq j'}|\theta_{j'k}|}=\dfrac{1}{\lam_2 \sum_{ k \neq j}|\theta_{jk}|}(X^\t_j z+X^\t_{j'} z),$\\
$H^*_2 = \dfrac{\sum_{ k \neq j}\theta_{jk}\hat \beta_k}{\sum_{ k \neq j}|\theta_{jk}|}+\dfrac{\sum_{ k \neq j'}\theta_{j'k}\hat \beta_k}{\sum_{ k \neq j'}|\theta_{j'k}|}=\dfrac{\theta_{jj'}}{\sum_{ k \neq j}|\theta_{jk}|} \ (\hat \beta_j + \hat \beta_{j'}),$\\
$H^*_3 = \lam_1 \cdot \lam_2 \sum_{ k \neq j}|\theta_{jk}|\cdot \sign{\hat \beta_j} + \lam_1 \cdot \lam_2 \sum_{ k \neq j'}|\theta_{j'k}|\cdot \sign{\hat \beta_{j'}}=0.$\\
\end{tabular}
\end{center}
Combining $ H*_1 $, $ H*_2 $ and $ H*_3 $, the sum of $ \hat \beta_j $ and $ \hat \beta_{j'} $ becomes
\[ |\hat \beta_j - \hat \beta_{j'}| = \dfrac{1}{\lam_2 (\sum_{ k \neq j}|\theta_{jk}|- \theta_{jj'})}\big|(X_j + X_{j'})^\t z\big|. \]

As a contrast, if we consider the $ l_2 $ penalty, i.e. elastic net,  the form of its solution $ \hat \beta^\en_j $ will be
\[-X^{\t}_j (y - X \hat \beta^\en) + \lambda_1 s^\en_j + \lambda_2 \hat \beta^\en_j = 0,
\]
where $ s^\en_j $ denotes the sign of $ \hat \beta^\en_j $. Set $z^\en = y - X \hat \beta^\en $. We have
\[\hat \beta^\en_j = (X^\t_j z^\en -\lam_1 s^\en_j)/\lam_2.\]
When $ \theta_{jj'}>0 $, the difference between $ \hat \beta^\en_j $ and $\hat \beta^\en_{j'} $ is
\[|\hat \beta^\en_j - \hat \beta^\en_{j'}| = \dfrac{1}{\lam_2} |(X_j - X_{j'})^\t z^\en|.\]
When $ \theta_{jj'}<0 $, the sum of $ \hat \beta^\en_j $ and $\hat \beta^\en_{j'} $ is
\[|\hat \beta^\en_j + \hat \beta^\en_{j'}| = \dfrac{1}{\lam_2} |(X_j + X_{j'})^\t z^\en|.\]
\end{proof}

Next, we prove Theorem~\ref{thm 1} and Theorem~\ref{thm 2}. Set $u=\sqrt{n}(\hat \beta- \beta)$ and $W = X^\t\e/\sqrt n $. $ u_S$, $W_S$ and $u_{S^c}$, $W_{S^c}$ denote the partition of $u$ and $W$ indexed by the set $S$ and $S^c$. We first introduce the following Lemmas in preparation for the theorem proof.
\begin{lem}\label{lem Omega}
Suppose Assumption~\ref{assum3} holds and $X$ is distributed from $ \mathcal{N}(0,\Sigma) $. Then for $ \Theta = \Sigma^{-1} $,  under settings that $p=O(e^{n^{c_3}})$, $q=O(n^{c_1})$ where $0<c_1+c_3<1$ and  a positive constant $ K $, we have this probability at least $ 1- 1/p $:
\[ \|\hat \Theta - \Theta\|_\infty \leqslant K \cdot (\dfrac{\log p}{n})^{1/2} \rightarrow 0 \ \text{as} \ n \rightarrow \infty.\]
\end{lem}
The result in Lemma~\ref{lem Omega} can be obtained straightly from the Theorem 1 in \citet{ravikumar2011high}, which proposed an elementwise $l_\infty$-bound of $\hat\Theta$. Lemma~\ref{lem Omega} guarantees that we have the accurate estimator $\hat \Theta$ of $\Theta$ with high probability. According to above result, it is easily to obtain that with high probability $\mathop {\lim }\limits_{n \to \infty } \lam _{\max}(\hat \Gamma) = \lam _{\max}(\Gamma )$, where $\hat \Gamma=\hat D-\hat \Theta$, $\hat D=\text{diag}(\hat d_1,\dots,\hat d_p)$ and $\hat d_j = \sum\limits_{j'=1}^p|\hat \theta_{jj'}|$.

\begin{lem}\label{lem tail}
Set $\lam_1=4\sigma\sqrt{n\log{p}}$. We have with probability at least $ 1-1/p $,
\begin{align*}
2\|W\|_\infty\leqslant \lam_1/\sqrt{n}.
\end{align*}
\end{lem}
\begin{proof}[Proof of Lemma~\ref{lem tail}]
By $\lam_1=4 \sigma\sqrt{n\log{p}}$, we have
\begin{align*}
P(2\|W\|_\infty > \lam_1/\sqrt{n} ) \leqslant P(\|X^\t \e\|_\infty> 2\sigma \sqrt{\log p} ).
\end{align*}
Since $\e$ are $i.i.d$ Gaussian variables, by Markov's inequality, we have
\begin{align*}
P(2\|W\|_\infty > \lam_1/\sqrt{n}) \leqslant p\exp (-\dfrac{2\sigma^2 n\log {p}}{n\sigma^2})  = \dfrac{1}{p}.
\end{align*}
\end{proof}

\begin{lem}\label{lem u1}
Conditional on $ \{ 2\|W\|_\infty\leqslant \lam_1/\sqrt{n} \} $. If $\lam_1=4\sigma\sqrt{n\log{p}}$ and $\lam_2 \left\|\beta \right\|_2=\sqrt{q}\lam_1$, then we have with probability at least $ 1-1/p $ that
\begin{align*}
\|u_{S^c}\|_1 \leqslant 3\|u_S\|_1.
\end{align*}
\end{lem}

\begin{proof}[Proof of Lemma~\ref{lem u1}]
By the definition of $\hat\beta$, we have
\begin{align*}
\dfrac{1}{2}\|y-X\hat\beta\|^2_2+\lam_1\|\hat\beta\|_1+\dfrac{\lam_2}{2}\hat\beta^\t\hat\Gamma\hat\beta \leqslant  \dfrac{1}{2}\|y-X\beta\|^2_2+\lam_1\|\beta\|_1+\dfrac{\lam_2}{2}\beta^\t\hat\Gamma\beta.
\end{align*}
After rearranging, above inequality become
\begin{align*}
(\hat \beta-\beta)^\t X^\t X (\hat \beta-\beta) + 2\lam_1 \|\hat \beta \|_1+\lam_2(\hat \beta^\t \hat \Gamma \hat\beta-\beta^\t \hat \Gamma \beta) \leqslant 2\e^\t X (\hat \beta - \beta)+2\lam_1\|\beta\|_1.
\end{align*}
Note that
\begin{align*}
-\frac{\lam_1}{\sqrt n} \| u_S \|_1+\frac{\lam_1}{\sqrt n} \| u_{S^c} \|_1 \leqslant \lam_1 \|\hat \beta \|_1 - \lam_1 \|\beta\|_1
\end{align*}
and
\begin{align*}
\beta^\t \hat \Gamma \beta -\hat \beta^\t \hat \Gamma \hat\beta &=\beta^\t \hat \Gamma \beta- (\hat \beta -\beta +\beta)^\t \hat \Gamma (\hat \beta -\beta+\beta),\\
&=-(\hat \beta -\beta)^\t \hat \Gamma(\hat \beta -\beta)-2 \beta^\t \hat \Gamma(\hat \beta -\beta),\\
&\leqslant 2\|\hat \Gamma\|_\infty\|\beta\|_2\|\hat \beta -\beta\|_1.
\end{align*}
Combining the above inequalities, we have
\begin{align*}
-\frac{\lam_1}{\sqrt n} \| u_S \|_1+\frac{\lam_1}{\sqrt n} \| u_{S^c} \|_1\leqslant W^\t u+\frac{\lam_2\|\hat \Gamma\|_\infty}{\sqrt n} \|\beta\|_2\|u\|_1.
\end{align*}
In addition, conditional on $2\|W\|_\infty \leqslant \lam_1/\sqrt{n}$, we have
\begin{align*}
-\lam_1\|u_S\|_1+\lam_1\|u_{S^c}\|_1 \leqslant \frac{\lam_1}{2} \| u_S \|_1+\frac{\lam_1}{2} \| u_{S^c} \|_1+ \lam_2  \|\hat \Gamma\|_\infty \|\beta\|_2\|u\|_1, \end{align*}
By $\lam_2 \|\hat \Gamma\|_\infty \|\beta\|_2=o(\lam_1)$ and Lemma~\ref{lem tail}, it implies that with probability at least $ 1 - 1/p $,
\[\|u_{S^c}\|_1\leqslant 3 \|u_S\|_1.\]
\end{proof}


\begin{proof}[Proof of Theorem~\ref{thm 1}]
Let $F(\beta ) = \dfrac{1}{2}\|y-X\beta\|_2^2 + \lam_1\|\beta \|_1 + \dfrac{\lam_2}{2}\beta^\t\hat\Gamma \beta $.
By the definition of $\hat\beta$, we set
\begin{align}\label{eq V}
V(u): = F(\hat\beta) - F(\beta) \leqslant 0.
\end{align}
Hence we have
\begin{align*}V(u) = \dfrac{1}{2}(\|y-X\hat \beta\|_2^2-\|y-X\beta\|_2^2) + \lam_1\|\hat\beta\|_1-\lam_1\|\beta\|_1 + \dfrac{\lam_2}{2}\hat\beta^\t\hat \Gamma \hat\beta-\dfrac{\lam_2}{2}\beta ^\t\hat\Gamma \beta.
\end{align*}
Let
\begin{align*}
&H_1=\dfrac{1}{2}(\|y-X\hat\beta\|_2^2-\|y-X\beta\|_2^2) = \dfrac{1}{2} u^\t C u - u^\t W, \\
&H_2=\lam_1\|\hat\beta\|_1-\lam_1\|\beta\|_1, \\
&H_3=\dfrac{1}{2}(\lam_2\hat\beta^\t \hat\Gamma \hat\beta - \lam_2\beta ^\t \hat\Gamma \beta).
\end{align*}
Note that
\begin{align*}
\|\hat\beta\|_1 -\|\beta\|_1 \geqslant-\|\hat\beta_S-\beta_S\|_1+ \|\hat\beta_{S^c} - \beta_{S^c}\|_1.
\end{align*}
Hence we have
\begin{align*}	{H_2} \geqslant  - \frac{\lam_1}{\sqrt n }\|u_S\|_1 + \frac{\lam_1}{\sqrt n }\|u_{S^c}\|_1
\end{align*}
and
\begin{align*}
H_3 \geqslant \dfrac{\lam_2}{2}\lam _{\min}(\hat\Gamma)(\|\frac{1}{\sqrt n }u\|_2^2 +\frac{2}{\sqrt n}u^\t_S\beta_S ).
\end{align*}
Combining above inequalities, we have
\begin{align*}
V(u)& \geqslant \dfrac{1}{2}u^\t C u - u^\t W-\frac{\lam_1}{\sqrt n }\|u_S\|_1+\frac{\lam_1}{\sqrt n}\|u_{S^c}\|_1 \\
&~~~+ \dfrac{\lam_2}{2}\lam _{\min}(\hat\Gamma)(\frac{1}{n}\|u\|_2^2 + \frac{2}{\sqrt n}u_S^\t\beta_S).
\end{align*}
Since $ V(u) \leqslant 0 $, according to the Lemma~\ref{lem Omega}~-~\ref{lem u1} and Assumption~\ref{assum1}, \eqref{eq V} becomes
\begin{align*}
\dfrac{K_1}{2}\|u\|_2^2 - u_S^\t W_S-\frac{\lam_1}{\sqrt n}\|u_S\|_1 +\dfrac{\lam_2}{2}\lam _{\min }(\Gamma)\big(\frac{1}{n}\|u||_2^2 + \frac{2}{\sqrt n }u_S^\t\beta _S\big)+ \frac{\lam_1}{\sqrt n}\|u_{S^c}\|_1-2u_{S^c}^\t W_{S^c} \leqslant 0.
\end{align*}
Let
\begin{align*}
&E_1=\dfrac{K_1}\|u\|_2^2-u_S^\t W_S - \frac{\lam_1}{\sqrt n }\|u_S\|_1 + \dfrac{\lam_2}{2}\lam _{\min }(\Gamma)\big(\frac{1}{n}\|u\|_2^2 + \frac{2}{\sqrt n }u_S^\t\beta _S\big),\\
&E_2 = \frac{\lam_1}{\sqrt n }\|u_{S^c}\|_1 - u_{S^c}^\t W_{S^c}.
\end{align*}
According to Lemma~\ref{lem tail}, we have $\|2W\|_\infty < \lam_1/\sqrt n $ with probability $1$, hence $E_2>0$. In addition, $E_1$ satisfies
\begin{align*}E_1& \geqslant \dfrac{K_1}{2}\|u\|_2^2 - \|u_S\|_2\|W_S\|_2 -  \dfrac{\lam_1}{\sqrt n }\|u_S\|_1 + \dfrac{\lam_2}{2n}\lam _{\min }(\Gamma)\|u_S\|_2^2 - \dfrac{\lam_2}{\sqrt n}\lam _{\min }(\Gamma)\|u_S\|_2 \|\beta _S\|_2,\\
& \geqslant \|u_S\|_2\big\{\dfrac{1}{2}\big(K_1 + \dfrac{\lam_2}{n}\lam _{\min}(\Gamma)\big)\|u\|_2- \|W_S\|_2-\dfrac{\lam_1}{\sqrt n } \sqrt q -\dfrac{\lam_2}{\sqrt n}\lam _{\min}(\Gamma) \|\beta _S\|_2\big\},
\end{align*}
which implies that
\begin{align*}
\|u\|_2 \leqslant \dfrac{2}{K_1 + \frac{\lam_2}{n}\lam _{\min }(\Gamma)}\big\{\dfrac{\lam_1}{2\sqrt n}\sqrt q  +  \frac{\lam_1}{\sqrt n}\sqrt q  + \frac{\lam_2}{\sqrt n}\lam _{\min }(\Gamma)\|\beta _S\|_2 \big\}.\
\end{align*}
It can be written as
\begin{align*}
\|\hat\beta -\beta\|_2 \leqslant \dfrac{2}{K_1 + \frac{\lam_2}{n}\lam _{\min }(\Gamma)}\big\{\dfrac{3\lam_1}{2 n }\sqrt q +\frac{\lam_2}{n}\lam _{\min }(\Gamma)\|\beta\|_2 \big\}.
\end{align*}
Then by $\lam_1=4\sigma\sqrt{n\log p}$, $q=O(n^{c_1})$ and $\lam_2\|\beta\|_2\leqslant \sqrt{q} \lam_1$, set a constant $K_3=(2+\lam_{\min}(\Gamma))^{-1}K_1$, we have with probability at least $ 1 - 1/p $
\[\| \hat\beta-\beta \|_2 \leqslant \dfrac{8\sigma}{K_3} \sqrt{\dfrac{q \log p}{n}}.\]
\end{proof}

\begin{lem}\label{lem sign}
Suppose Assumption~\ref{assum3}~-~\ref{assum2} hold, then we have
\begin{align*}P( \hat S = S)  \geqslant P(A \cap B), \end{align*}
for
\begin{align*}
A=&\big\{\big|\big(C_S+\dfrac{ \lam_2}{n}\Gamma_S\big )^{-1}W_S\big| < \sqrt n |\beta_S|\notag\\
&-\big|\big(C_S+\dfrac{\lam_2}{n}\Gamma_S\big)^{-1}\big[\dfrac{ \lam_1}{\sqrt n}\sign{\beta_S}+\dfrac{\lam_2}{n}\Gamma_S\beta_S\big]\big| \big\},\\
B = &\big\{\big |\big(C_{S^c} + \dfrac{\lam_2}{n}\Gamma_{S^c}\big)\big(C_S+\dfrac{\lam_2}{n}\Gamma_S\big)^{-1}W_S-W_{S^c}\big| \leqslant \dfrac{\lam_1}{\sqrt n}\alpha_1 \big\},
\end{align*}
where $ C_S = \dfrac{1}{n}X^\t_S X_S $ and $ C_{S^c}=\dfrac{1}{n}X^\t_{S^c}X_S $.
\end{lem}
\begin{proof}[Proof of Lemma~\ref{lem sign}]
By KKT (Karush-Kuhn-Tucker) conditions for optimality in convex problem, the point $\hat\beta$ is optimal if and only if
\begin{align*}
X^\t X\hat\beta-X^\t y + \lam_2\hat \Gamma \hat \beta + \lam_1 \gamma = 0,
\end{align*}
\begin{align*}
\gamma_j \in \left\{
\begin{array}{*{20}{l}}
\{\sign{\hat\beta_j}\}, & \hat \beta_j \ne 0,\\
{ [ - 1,1],}&\text{otherwise}.
\end{array} \right.
\end{align*}
$\hat\beta$ is sign consistency if and only if $\hat\beta_{S^c}=0$, $\hat\beta_S\ne 0$. Thus we have sign consistency holds if and only if following hold:
\begin{align}\label{eq s1}
(X_S^\t X_S + \lam_2\hat \Gamma_S)\hat \beta_S - X_S^\t X_S\beta_S - X_S^\t \e = -\lam_1\gamma_S,
\end{align}
\begin{align}\label{eq s2}
\big|(X_{S^c}^\t X_S + \lam_2\hat\Gamma_{S^c})\hat\beta_S - X_{S^c}^\t X_S\beta_S - X_{S^c}^\t \e \big| \leqslant |\lam_1|,
\end{align}
\begin{align}\label{eq u}
|\hat \beta_S- \beta_S| < |\beta_S|.
\end{align}
By~\eqref{eq s1} - \eqref{eq u}, according to Lemma~\ref{lem Omega} and applying the nontations that $ C_S = \dfrac{1}{n}X^\t_S X_S $, $ C_{S^c}=\dfrac{1}{n}X^\t_{S^c}X_S $, when $ n $ is large, we have
\begin{align}\label{eq 01}
&\big|\big(C_S +\dfrac{\lam_2}{n}\Gamma_S\big)^{-1}W_S\big|\\ \notag
&<\sqrt n
|\beta_S|-\big|\big(C_S +\dfrac{\lam_2}{n}\Gamma_S\big)^{-1}\big[\dfrac{\lam_1}{\sqrt n}\sign{\beta_S}+\dfrac{\lam_2}{\sqrt n}\Gamma_S\beta_S\big]\big|
\end{align}
and
\begin{align}\label{eq 02}
&\big |\big(C_{S^c} + \dfrac{\lam_2}{n}  \Gamma_{S^c}\big)\big(C_S+\dfrac{\lam_2}{n} \Gamma_S\big)^{-1}W_S-W_{S^c}\big|\\ \notag
&\leqslant\dfrac{\lam_1}{\sqrt n}\big(1-\big|\big(C_{S^c} + \dfrac{\lam_2}{n} \Gamma_{S^c}\big)\big(C_S+\dfrac{\lam_2}{n} \Gamma_S\big)^{-1}\big[\sign{\beta_S}+\dfrac{\lam_2}{\lam_1} \Gamma_S\beta_S\big]-\dfrac{ \lam_2}{\lam_1} \Gamma_{S^c}\beta_S\big|\big).
\end{align}
According to Assumption~\ref{assum2}, above inequality~\eqref{eq 02} can be written as
\begin{align}\label{eq b}
\big |\big(C_{S^c} + \dfrac{\lam_2}{n}  \Gamma_{S^c}\big)\big(C_S+\dfrac{\lam_2}{n} \Gamma_S\big)^{-1}W_S-W_{S^c}\big| \leqslant \dfrac{ \lam_1}{\sqrt n}\alpha_1.
\end{align}
\eqref{eq 01} and~\eqref{eq b} coincide with $A$ and $B$ respectively.
\end{proof}


\begin{proof}[Proof of Theorem~\ref{thm 2}]
According to Lemma~\ref{lem sign}, we have
\begin{align*}
P\{\hat S = S\}  \geqslant P\{A\cap B\}
\end{align*}
and
\begin{align*}
1-P(A \cap B) &\leqslant P(A^c) + P(B^c),\\
& \leqslant \sum\limits_{i \in S} P(|z_i| \geqslant \sqrt n |\beta _i| - |b_i|) + \sum\limits_{i \in S^c} P\big(|\xi _i| \geqslant \frac{\lam_1}{\sqrt n}\alpha_1\big),
\end{align*}
where $z = (z_1,\ldots,z_q)^\t = (C_S + \frac{\lam_2}{n}\Gamma_S)^{-1}W_S$, $\xi  = (\xi_1,\ldots,\xi _{p-q})^\t = W_{S^c} - (C_{S^c}+\frac{\lam_2}{n}\Gamma_{S^c})(C_S + \frac{\lam_2}{n}\Gamma_S)^{-1}W_S$ and $b = (b_1,\ldots,b_q)^\t = (C_S + \frac{\lam_2}{n}\Gamma_S)^{-1}\big[\frac{\lam_1}{\sqrt n}\gamma_S + \frac{\lam_2}{n}\Gamma_S\beta_S\big]$.
Since $\lam_1=4\sigma\sqrt{n\log{p}}$, $\|\beta\|_2 \leqslant \sqrt{q}\lam_1/\lam_2$ and $q=O(n^{c_1})$ where $ c_1+c_3 < c_2 <1$, we have
\begin{align*}
|b_i| &\leqslant \big|\big(C_S + \frac{\lam_2}{n}\Gamma_S\big)^{-1}\frac{\lam_1}{\sqrt n}\gamma_S \big|_i + \big|\big(C_S + \frac{\lam_2}{n}\Gamma_S\big)^{-1}\frac{\lam_2}{n}\Gamma_S
\beta_S \big|_i,\\& \leqslant \frac{\lam_1}{\sqrt n}\frac{\sqrt q}{\lam _{\min }\big(C_S + \frac{\lam_2}{n}\Gamma_S\big)} + \frac{\lam_2}{n}\frac{\lam _{\max }(\Gamma_S) \|\beta_S\|_\infty}{\lam _{\min }\big(C_S + \frac{\lam_2}{n}\Gamma_S\big)},\\
& = o(n^{\frac{c_2}{2}}).
\end{align*}
Since $z_i$, $\xi_i $ are Gaussian variables with mean zero and finite variance. For $s>0$, the Gaussian distribution has its tail probability bound, we have
\begin{align*}
\sum\limits_{i\in S} P(|z_i| \geqslant \sqrt n |\beta _i|-|b_i|)  \leqslant q \cdot O\big\{ 1 - \Phi \big((1 + o(1))\frac{K_2}{s}n^{\frac{c_2}{2}}\big)\big\}  = o({e^{-n^{c_3}}})
\end{align*}
and
\begin{align*}
\sum\limits_{i\in S^c} P\big(|\xi _i| \geqslant \frac{\lam_1}{2\sqrt n}\alpha_1\big) \leqslant (p-q) \cdot O\big\{1-\Phi \big(\frac{1}{s}\frac{\lam_1}{2\sqrt n}\alpha_1\big)\big\}  =o(e^{-n^{c_3}}).
\end{align*}
Then
\begin{align*}
P(\hat S=S)\geqslant 1-1/p\rightarrow 1, \ as \ n\rightarrow\infty.
\end{align*}
\end{proof}

\bibliographystyle{apalike}
\bibliography{C:/reference/reference}
\end{document}